\documentclass[a4paper,UKenglish,cleveref, autoref, thm-restate]{lipics-v2021}

\usepackage{url}
\usepackage{amsmath,amssymb,amsthm,hyperref}
\usepackage{color}
\usepackage{amsfonts}
\usepackage{algpseudocode}
\usepackage[linesnumbered,ruled,vlined]{algorithm2e}
\SetKwProg{Procedure}{Procedure}{}{}

\usepackage{comment}

\newcommand{\E}{\mathop{\mathbb{E}}}
\newcommand{\EM}{\mathrm{EM}}
\newcommand{\TV}{\mathrm{TV}}

\newcommand{\OPT}{\mathrm{opt}}
\newcommand{\FOPT}{\mathrm{fopt}}

\title{Average sensitivity of the Knapsack Problem} 

\author{Soh Kumabe}{The University of Tokyo}{soh\_kumabe@mist.i.u-tokyo.ac.jp}{}{Supported by JST, PRESTO Grant Number JPMJPR192B.}

\author{Yuichi Yoshida}{National Institute of Informatics}{yyoshida@nii.ac.jp}{https://orcid.org/0000-0001-8919-8479}{Supported by JST, PRESTO Grant Number JPMJPR192B.}

\authorrunning{S. Kumabe and Y. Yoshida}

\Copyright{Soh Kumabe and Yuichi Yoshida}



\category{} 

\relatedversion{} 

\nolinenumbers 

\EventEditors{John Q. Open and Joan R. Access}
\EventNoEds{2}
\EventLongTitle{42nd Conference on Very Important Topics (CVIT 2016)}
\EventShortTitle{CVIT 2016}
\EventAcronym{CVIT}
\EventYear{2016}
\EventDate{December 24--27, 2016}
\EventLocation{Little Whinging, United Kingdom}
\EventLogo{}
\SeriesVolume{42}
\ArticleNo{23}

\begin{document}

\maketitle

\begin{abstract}

In resource allocation, we often require that the output allocation of an algorithm is stable against input perturbation because frequent reallocation is costly and untrustworthy.
Varma and Yoshida (SODA'21) formalized this requirement for algorithms as the notion of average sensitivity.
Here, the average sensitivity of an algorithm on an input instance is, roughly speaking, the average size of the symmetric difference of the output for the instance and that for the instance with one item deleted, where the average is taken over the deleted item.

In this work, we consider the average sensitivity of the knapsack problem, a representative example of a resource allocation problem.
We first show a $(1-\epsilon)$-approximation algorithm for the knapsack problem with average sensitivity $O(\epsilon^{-1}\log \epsilon^{-1})$.
Then, we complement this result by showing that any $(1-\epsilon)$-approximation algorithm has average sensitivity $\Omega(\epsilon^{-1})$.
As an application of our algorithm, we consider the incremental knapsack problem in the random-order setting, where the goal is to maintain a good solution while items arrive one by one in a random order.
Specifically, we show that for any $\epsilon > 0$, there exists a $(1-\epsilon)$-approximation algorithm with  amortized recourse $O(\epsilon^{-1}\log \epsilon^{-1})$ and amortized update time $O(\log n+f_\epsilon)$, where $n$ is the total number of items and $f_\epsilon>0$ is a value depending on $\epsilon$.
\end{abstract}
\newpage



\section{Introduction}

\subsection{Background and Motivation}

The \emph{knapsack problem} is one of the most fundamental models in \emph{resource allocation}, which handles the selection of good candidates under a budget constraint. For example, it can model the hiring process of employees in a company and the selection process of government projects.
The knapsack problem is formally defined as follows.
The input is a pair $(V,W)$, where $V$ is a set of $n$ items, $W\in \mathbb{R}_{>0}$ is a weight limit, and each item $i\in V$ has a positive weight $w(i)\leq W$ and value $v(i)$.
The goal of the problem is to find a set of items $S \subseteq V$ that maximizes the total value $\sum_{i\in S}v(i)$, subject to the weight constraint $\sum_{i\in S}w(i)\leq W$.


Sometimes, the information used to allocate resources is uncertain or outdated. For example, suppose that a satellite isolated from the Earth is taking actions. The satellite has a list of potential actions $V$, and each action has a fixed value $v(i)$ and fuel consumption $w(i)$. Some of the actions may be infeasible at the moment due to the satellite's conditions, such as the atmosphere or surrounding space debris. The satellite's objective is to find a combination of feasible actions with the highest possible total value without running out of $W$ amount of fuel.

The control room on the Earth wants to know the satellite's action for future planning. Since direct communication to the satellite is not always possible, the control room would simulate the satellite's decision process. However, the list $V$ that the control room has as potential actions for the satellite may be different from the list that the satellite uses. Even in such a situation, the control room would like to predict the actual actions taken by the satellite with some accuracy. Such a purpose can be achieved by designing algorithms with small \emph{average sensitivity}~\cite{Murai2019,Varma2021}.

The following definitions are necessary to formally define the average sensitivity.
Let $V$ be a finite set.
Then, the \emph{Hamming distance}\footnote{The Hamming distance is usually defined for strings, but our definition matches the formal definition by considering a set $S \subseteq V$ as a binary string of length $|V|$.} between two sets $S,S' \subseteq V$ is $|S\triangle S'|$, where $S\triangle S'=(S\setminus S')\cup (S'\setminus S)$ is the \emph{symmetric difference}.
For two probability distributions $\mathcal{S}$ and $\mathcal{S}'$ over sets in $V$, the \emph{earth mover's distance} $\EM(\mathcal{S},\mathcal{S}')$ between them is given by
\begin{align}
  \min_{\mathcal{D}} \E_{(S,S') \sim \mathcal{D}} |S \triangle S'|,
  \label{eq:earth-mover}
\end{align}
where the minimum is taken over distributions of a pair of sets such that its marginal distributions on the first and the second coordinates are equal to $\mathcal{S}$ and $\mathcal{S}'$, respectively.
Let $\mathcal{A}$ be a (randomized) algorithm for the knapsack problem, and let $(V,W)$ be an instance of the knapsack problem.
We abuse the notation $\mathcal{A}$ (resp., $\mathcal{A}^i$) to represent the output (distribution) of the algorithm $\mathcal{A}$ on the instance $(V,W)$ (resp., $(V\setminus \{i\},W)$).
Then, the \emph{average sensitivity} of the algorithm (on the instance $(V,W)$) is
\begin{align}
    \frac{1}{n}\sum_{i\in V}\EM(\mathcal{A},\mathcal{A}^i).
  \label{eq:average-sensitivity}
\end{align}
An algorithm is informally called \emph{stable-on-average} if its average sensitivity is small.

\subsection{Our Contributions}

Our main contribution is a fully polynomial-time randomized approximation scheme (FPRAS) for the knapsack problem with a small average sensitivity:
\begin{theorem}\label{thm:intro}
  For any $\epsilon>0$, there is a $(1-\epsilon)$-approximation algorithm for the knapsack problem with time complexity polynomial in $\epsilon^{-1}n$ and average sensitivity $O(\epsilon^{-1}\log \epsilon^{-1})$, where $n$ is the number of items.
\end{theorem}
It is noteworthy that Kumabe and Yoshida~\cite{kumabe2022average} have presented a $(1-\epsilon)$-approximation algorithm with $O(\epsilon^{-1}\log^3 (nW))$ average sensitivity for the knapsack problem with integer weights. In contrast, our algorithm can be applied to an instance with non-integer weight and has a smaller average sensitivity bound independent of the weight limit $W$.



The following proposition states that the algorithm that outputs an optimal solution has an unbounded average sensitivity.
Thus to bound the average sensitivity, we should allow for approximation as in Theorem~\ref{thm:intro}.
\begin{proposition}\label{prop:unbounded}
    The average sensitivity of the algorithm that outputs an optimal solution can be as large as $\Omega(n)$.
\end{proposition}

Let $\OPT(V)$ be the optimal value of the original instance.
The basic idea behind the algorithm of Theorem~\ref{thm:intro} is classifying items into two categories according to whether their values are more than a threshold $\approx \epsilon\cdot \OPT(V)$. 
It is not difficult to design stable-on-average $(1-\epsilon)$-approximation algorithms when all the items belong to one of the two categories.
Indeed, if all the items are small, i.e., with values at most $\epsilon \cdot \OPT(V)$, then we can prove that a variant of the greedy algorithm has an approximation ratio $(1-O(\epsilon))$ and an $O(\epsilon^{-1})$ average sensitivity.
If all items are large, i.e., with values at least $\epsilon \cdot \OPT(V)$, then any algorithm has a small average sensitivity because any feasible solution contains at most $\epsilon^{-1}$ items.
To combine these two algorithms, we add a subset of the large items selected by the \emph{exponential mechanism}~\cite{mcsherry2007mechanism} to the output.
However, if we apply the exponential mechanism to all possible sets of large items, the average sensitivity becomes too high because the number of such sets is exponentially large.
Therefore, we need to reduce the number of candidate sets appropriately without sacrificing the approximation guarantee.
A more detailed overview is provided in Section~\ref{sec:overview}.

Next, we show that our upper bound on the average sensitivity is tight up to a logarithmic factor in $\epsilon^{-1}$.
\begin{theorem}\label{thm:intro-lower-bound}
  Let $\epsilon > 0$.
  Then any $(1-\epsilon)$-approximation algorithm for the knapsack problem has average sensitivity $\Omega(\epsilon^{-1})$.
\end{theorem}

The \emph{simple knapsack problem} is a special case of the general knapsack problem in which the value of each item is proportional to its weight.
For this problem, we obtain a deterministic algorithm with a better average sensitivity:
\begin{theorem}\label{thm:intro-simple}
  For any $\epsilon>0$, there is a deterministic $(1-\epsilon)$-approximation algorithm for the simple knapsack problem with time complexity polynomial in $\epsilon^{-1}n$ and average sensitivity $O(\epsilon^{-1})$, where $n$ is the number of items.
\end{theorem}

Finally, we discuss the connection to dynamic algorithms. In the \emph{incremental dynamic knapsack problem}, items arrive one by one, and the goal is to maintain an approximate solution for the current set of items.
The amortized recourse of a dynamic algorithm is the average Hamming distance between the outputs before and after an item arrives.
More formally, the \emph{amortized recourse} of a deterministic algorithm over a stream $v_1,\dots, v_n$ of items is defined as
$\frac{1}{n}\sum_{i=1}^{n}|X_{i-1}\triangle X_i|$,
where $X_i$ is the solution of the algorithm right after $v_i$ is added.
The amortized recourse of a randomized algorithm is the expectation of amortized recourse over the randomness of the algorithm.
When the arrival order is random, we can construct an algorithm for the incremental dynamic knapsack problem using our stable-on-average algorithm:
\begin{theorem}\label{thm:dynamic}
For any $\epsilon>0$, there exists $f_\epsilon>0$ such that there is a $(1-\epsilon)$-approximation algorithm with amortized recourse $O(\epsilon^{-1}\log \epsilon^{-1})$ and update time $O(f_\epsilon+\log n)$ for the incremental knapsack problem in the random-order setting.
\end{theorem}
The fact that the output of our stable-on-average algorithm does not depend on the arrival order of the items but depends only on the current set of items implies that this result can also be applied to the \emph{decremental knapsack problem}, in which we are to maintain a good solution while the items are removed one by one from the initial set of items.

\subsection{Related Work}

\subsubsection{Knapsack Problem}

The \emph{knapsack problem} is one of the $21$ problems that was first proved to be NP-hard by Karp~\cite{karp1972reducibility}, but admits a pseudo-polynomial time algorithm via dynamic programming~\cite{kellerer2004knapsack}. Ibarra and Kim~\cite{ibarra1975fast} proposed the first fully polynomial-time approximation scheme (FPTAS) for the knapsack problem. The main idea is to round the values of items into multiples of a small value, and run a dynamic programming algorithm on the resulting instance. Since then, several faster algorithms have been developed~\cite{chan2018approximation,jin2019improved,kellerer1999new,kellerer2004improved,lawler1979fast,rhee2015faster}.

The knapsack problem has been studied in the context of online settings.
In the most basic \emph{online knapsack problem}, items arrive one by one, and when an item arrives, it is irrevocably added it to the knapsack or discarded. The difference from our dynamic model is that, in the online knapsack problem, once an item is accepted or rejected, the decision cannot be reverted.
However, this is not the case in our model.
While Marchetti-Spaccamela and Vercellis~\cite{marchetti1995stochastic} showed that no algorithm has a bounded competitive ratio for this problem in general, Buchbinder and Naor~\cite{buchbinder2009online} proposed an $O(\log(U/L))$-competitive algorithm for the case where all weights are much smaller than the weight limit, where $U$ and $L$ are the upper and lower bounds on the \emph{efficiency} $w(i)/v(i)$. Zhou, Chakrabarty, and Lukose showed that this bound is tight~\cite{zhou2008budget}.
Several variations of the online knapsack problem have been investigated, such as the \emph{removable online knapsack problem}~\cite{iwama2002removable}, which allows removing items added from the knapsack, the \emph{knapsack secretary problem}~\cite{babaioff2007knapsack}, which concerns items arriving in random order, and the \emph{online knapsack problem with a resource buffer}~\cite{han2019online}, in which we have a buffer that can contain higher weights of items than the weight limit, and at the end, we can select a subset of the stored items as the output.

The problem most closely related to ours is the \emph{dynamic setting}.
In this problem, the goal is to maintain a good solution while items are added or deleted dynamically.
In this model, there is no restriction on adding items that were previously deleted from the knapsack (and vice versa).
While designing fast algorithms that maintain an exact solution is a major issue of focus in the context of data structures~\cite{eppstein1999dynamic}, the problem of maintaining approximate solutions has also been investigated in a variety of problems including the shortest path problem~\cite{shiloach1981line,van2019dynamic}, maximum matching problem~\cite{gupta2013fully,onak2010maintaining}, and set cover problem~\cite{gupta2017online}.
Recently, dynamic approximation algorithms with bounded recourse have been studied for several problems including the bin-packing problem~\cite{feldkord2018fully}, set cover problem~\cite{gupta2017online}, and submodular cover problem~\cite{gupta2020fully}.
As for the knapsack problem, a recent work of B{\"o}hm et al.~\cite{bohm2020fully} presented an algorithm that maintains a $(1-\epsilon)$-approximate solution with polylogarithmic update time in the fully dynamic setting, where both additions and deletions of items are allowed.
Because this model considers both additions and deletions, the solution must be drastically changed even after a single update operation. Therefore, they proposed an algorithm that maintains a data structure from which a good solution can be recovered (with a time proportional to the size of the output), rather than explicitly maintaining the solution as a set of items.

\subsubsection{Average Sensitivity}
Murai and Yoshida~\cite{Murai2019} introduced the notion of average sensitivity for centralities, i.e., importance of nodes and edges, on networks to compare various notions of centralities.
The notion of average sensitivity for graph problems was recently introduced by Varma and Yoshida~\cite{Varma2021}, in which they studied various graph problems, including the minimum spanning tree problem, minimum cut problem, maximum matching problem, and minimum vertex cover problem.
Zhou and Yoshida~\cite{Yoshida2021} presented a $(1-\epsilon)$-approximation algorithm for the maximum matching problem with sensitivity solely depending on $\epsilon$, where the \emph{(worst-case) sensitivity} is defined as~\eqref{eq:average-sensitivity} with the average replaced by the maximum over $i$.
Kumabe and Yoshida~\cite{kumabe2022average} presented a stable-on-average algorithm for the maximum weight chain problem on directed acyclic graphs. The approximation ratio of their algorithm is $(1-\epsilon)$ and average sensitivity is polylogarithmic in the number of vertices in the graph, which roughly corresponds to the number of states of the dynamic programming.
They designed stable-on-average algorithms for a wide range of dynamic programming problems, including the knapsack problem with integer weights, by reducing them to the maximum weight chain problem.
Peng and Yoshida~\cite{Peng2020} analyzed the average sensitivity of spectral clustering and showed that it is proportional to $\lambda_2/\lambda_3^2$, where $\lambda_i$ is the $i$-th smallest eigenvalue of the (normalized) Laplacian of the input graph.
Intuitively, this bound indicates that spectral clustering is stable on average when the input graph can be well partitioned into two communities but not into three.
This implies that spectral clustering is stable on average at relevant instances.

\subsubsection{Differential Privacy}

\emph{Differential privacy}, proposed by Dwork \emph{et al.}~\cite{dwork2006calibrating} is related to average sensitivity.
For $\epsilon>0$ and $\delta>0$, an algorithm is \emph{$(\epsilon,\delta)$-differential private} if for any set $S$ and a pair of two inputs $V$ and its ``neighbor'' $V'$, we have
\begin{align*}
    \Pr[\mathcal{A}(V)\in S]\leq \exp(\epsilon)\Pr[\mathcal{A}(V')\in S]+\delta,
\end{align*}
where $\mathcal{A}(V)$ is the output distribution of the algorithm for the input $V$.
In the context of the knapsack problem, it is natural to regard that $V$ and  $V'$ are neighbors when $V'$ is obtained from $V$ by deleting one item (or vice versa).

It is known that $(\epsilon,\delta)$-differential privacy implies that the total variation distance between $\mathcal{A}(V)$ and $\mathcal{A}(V')$ is at most $\epsilon+\delta$~\cite{vadhan2017complexity}.
Here, the \emph{total variation distance} $\TV(\mathcal{S},\mathcal{S}')$ between two distributions $\mathcal{S},\mathcal{S}'$ is given by $
  \min_{\mathcal{D}} \E_{(S,S') \sim \mathcal{D}} 1_{S=S'}$,
where $1_P$ is the indicator of a predicate $P$, i.e., it is $1$ if $P$ holds and $0$ otherwise, and the minimum is taken over distributions of a pair of sets such that its marginal distributions on the first and second coordinates are equal to $\mathcal{S}$ and $\mathcal{S}'$, respectively. The total variation distance differs from the earth mover's distance in that we measure the distance of sets by the trivial zero-one distance, and not by the Hamming distance.

\subsection{Organization}

The rest of this paper is organized as follows. In~\cref{sec:prelim}, we review the basics of the knapsack problem, especially focusing on the fractional knapsack problem.
In~\cref{sec:small-weight}, we present a stable-on-average $(1-O(\epsilon))$-approximation algorithm when all items have value at most $\epsilon$.
In~\cref{sec:general}, we present an (inefficient) stable-on-average algorithm for the general setting. We discuss on improving the time complexity to obtain an FPRAS in~\cref{sec:fptas}.
In~\cref{sec:lower_bound}, we prove the $O(\epsilon^{-1})$ lower bound on the average sensitivity of $(1-\epsilon)$-approximation algorithms.
In~\cref{sec:simple}, we present a deterministic $(1-\epsilon)$-approximation algorithm with $O(\epsilon^{-1})$ sensitivity for the simple knapsack problem.
Finally, in~\cref{sec:recourse}, we discuss the application to the dynamic knapsack problem in the random-order setting.



\section{Basic Facts about the Knapsack Problem}\label{sec:prelim}

For an instance $(V,W)$ of the knapsack problem, let $\OPT(V,W)$ be the optimal value for the instance.
As $W$ is often set to $1$, we define $\OPT(V) := \OPT(V,1)$ for convenience.
In this work, we assume that each item has a unique (comparable) identifier that remains unchanged with the deletion of other items. This naturally defines the ordering of items. In our algorithms, we implicitly use this ordering for the tiebreak. For example, when we sort items, we use a sorting algorithm that is stable with respect to this ordering.
For a set $S \subseteq V$, let $v(S)$ and $w(S)$ denote $\sum_{i \in S}v(i)$ and $\sum_{i \in S}w(i)$, respectively.


We consider the following fractional relaxation of the knapsack problem, which we call the \emph{fractional knapsack problem}:
\begin{align*}
    \begin{array}{lll}
    \text{maximize} & \displaystyle \sum_{i\in V}v(i)x(i), \\
    \text{subject to} & \displaystyle \sum_{i\in V}w(i)x(i)\leq W,\\
    & 0\leq x(i)\leq 1\quad (i\in V).
    \end{array}
\end{align*}
Then, we denote the optimal value of this problem by $\FOPT(V,W)$.
We also define $\FOPT(V) := \FOPT(V,1)$, and the \emph{efficiency} of an item $i$ is $v(i)/w(i)$.
The following is well known:
\begin{lemma}[\cite{dantzig1957discrete}]\label{lem:fractional-optimal-solution}
    Suppose that items in $V$ are sorted in non-increasing order of efficiency, i.e., $v(1)/w(1)\geq \cdots \geq v(n)/w(n)$.
    Let $k$ be the largest index with $w(1)+\cdots+w(k)\leq 1$.
    Then, $\FOPT(V,1)$ is achieved by the solution
    \[
        x(i)=\begin{cases}
            1 & (i=1,\dots, k), \\
            \frac{1-\sum_{j=1}^{k}w(j)}{w(i+1)} &  (i=k+1),\\
            0 & (i \geq k+2).
        \end{cases}
    \]
\end{lemma}
Using Lemma~\ref{lem:fractional-optimal-solution}, $\FOPT(V,W)$ can be computed in polynomial time.
The following is a direct consequence of Lemma~\ref{lem:fractional-optimal-solution}:
\begin{lemma}\label{lem:frac_efficient}
    For $W\leq 1$, we have $\FOPT(V,W)\geq W\cdot \FOPT(V)$.
\end{lemma}

The followings are also folklore results.
\begin{lemma}\label{lem:opt2fopt}
    $\OPT(V)\leq \FOPT(V)\leq 2\OPT(V)$ holds.
\end{lemma}
\begin{proof}
$\OPT(V)\leq \FOPT(V)$ is clear.
Now we show the other inequality.
Let $k$ be the largest index with $w(1)+\dots+w(k)\leq 1$. If $k=n$, the statement is clear. Otherwise, we have $\FOPT(V)\leq (v(1)+\dots v(k))+v(k+1)\leq \OPT(V)+\OPT(V)=2\OPT(V)$.
\end{proof}

\begin{lemma}\label{lem:fracdiffsum}
For an item set $U$ and weight limit $W$, we have
\[
    \sum_{i\in U}(\FOPT(U,W)-\FOPT(U\setminus \{i\},W))\leq \FOPT(U,W).
\]
\end{lemma}
\begin{proof}
Let $x \in \mathbb{R}^U$ be an optimal solution to the fractional knapsack problem for the instance $(U,W)$.
For $i\in U$, we have $\FOPT(U\setminus \{i\},W)\geq \FOPT(U,W)-x(i)v(i)$ because $x-x(i)e_i$ (regarded as a vector in $\mathbb{R}^{U \setminus \{i\}}$) is a feasible solution to the instance $(U\setminus \{i\},W)$, where $e_i$ is the characteristic vector of $i$.
Therefore, we have
\[
    \sum_{i\in U}(\FOPT(U,W)-\FOPT(U\setminus \{i\},W))
    \leq \sum_{i\in U}x(i)v(i) = \FOPT(U,W).
    \qedhere
\]
\end{proof}

\section{Items with Small Values}\label{sec:small-weight}

In this section, we provide a stable-on-average algorithm for instances with each item having a small value.
Specifically, we show the following:
\begin{theorem}\label{thm:small}
    For any $\epsilon,\delta > 0$, there exists an algorithm for the knapsack problem with average sensitivity $O(\epsilon^{-1})$ that, given an instance $(V,W)$ with each item having value at most $\delta$, outputs a solution with value at least $(1-\epsilon)\FOPT(V,W)-\delta$.
\end{theorem}
This algorithm will be used in our algorithm for the general case as a subroutine in Section~\ref{sec:general}.
Note that it suffices to design an algorithm for the case in which the weight limit is $1$ because we can first divide the weight of each item by $W$ and feed the resulting instance to the algorithm.
Hence in the rest of this paper, we assume that the weight limit is $1$ (and use the symbol $W$ for a different purpose).
A pseudocode of our algorithm is given in Algorithm~\ref{alg:smallknapsack}.

\begin{algorithm}[t!]
\caption{A $(1-O(\epsilon))$-approximation algorithm with small average sensitivity}\label{alg:smallknapsack}
\Procedure{\emph{\Call{ModifiedGreedy}{$\epsilon,V$}}}{
    Reorder the items of $V$ so that $v(1)/w(1)\geq \cdots \geq v(n)/w(n)$\;
    Let $W$ be sampled from $[1-\epsilon,1]$ uniformly at random\;
    Let $t$ be the maximum index with $\sum_{i=1}^{t}w(i)\leq W$\;\label{line:smalltaket}
    \Return $\{1,\dots, t\}$.
}
\end{algorithm}


\subsection{Modified Greedy Algorithm}\label{sec:greedy}
We first consider the following greedy algorithm:
Reorder the items in $V$ in a non-increasing order of efficiency, i.e., $v(1)/w(1)\geq \cdots \geq v(n)/w(n)$, take the maximum index $t$ with $\sum_{i=1}^{t}w(i)\leq 1$, and then output the set $\{1,\dots, t\}$.
It is easy to verify that this natural algorithm outputs a solution with value at least $\OPT(V)-\delta$ when every item has value at most $\delta$.
However, its average sensitivity is unbounded.
To confirm this, we consider the proof of Proposition~\ref{prop:unbounded}.
\begin{proof}[Proof of Proposition~\ref{prop:unbounded}]
Let $k$ be an integer.
Let $V_1$ be the set of $k$ copies of an item with weight $\frac{1}{k}$ and value $\frac{1}{k}$, $V_2$ be the set of $k$ copies of an item with weight $\frac{1}{k^2}$ and value $\frac{1}{k^3}$, and let $V=V_1\cup V_2$.
The optimal solution for the instance $V$ is $V_1$, which does not change upon deleting an item in $V_2$.
However, if an item $i\in V_1$ is deleted, the optimal solution changes to $V\setminus \{i\}$. Therefore, the average sensitivity is
\[
    \frac{k}{2k}\cdot 0 + \frac{k}{2k} \cdot (k+1) =\frac{k+1}{2},
\]
which is as large as $\Omega(n)$.
\end{proof}
This proof also shows that for any $\delta>0$, the natural greedy algorithm has an unbounded average sensitivity even if the values of items are at most $\delta$, by setting $k\geq \lceil \delta^{-1} \rceil$.

An issue of the greedy algorithm is that, when there are many small items (in weight) ``right after'' the weight limit in the efficiency ordering (e.g., the items in $V_2$ in the instance used in the proof of Proposition~\ref{prop:unbounded}), deleting an item in the current solution may cause the algorithm to add many of the small items.
To resolve this issue, we uniformly sample the weight limit $W$ from $[1-\epsilon,1]$, instead of fixing it to $1$.

\subsection{Analysis}
First we analyze the approximation ratio of Algorithm~\ref{alg:smallknapsack}.
For $W\in [1-\epsilon,1]$, let $S(W)$ (resp., $S^i(W)$) be the output of the greedy algorithm run on $V$ (resp., $V\setminus \{i\}$) with the weight limit set to $W$.
Note that Algorithm~\ref{alg:smallknapsack} run on $V$ (resp., $V \setminus \{i\}$) first samples $W$ from $[1-\epsilon,1]$ uniformly at random and then outputs $S(W)$ (resp., $S^i(W)$).
The following lemma is obtained from a standard argument about the knapsack problem.
\begin{lemma}\label{lem:smallapprox}
Suppose every item has value at most $\delta$.
Then, Algorithm~\ref{alg:smallknapsack} outputs a solution with value at least $(1-\epsilon)\FOPT(V)-\delta$.
\end{lemma}
\begin{proof}
We have
\[
    \sum_{i=1}^{t}w(i)\geq \FOPT(V,W)-\delta \geq W\FOPT(V,1)-\delta \geq (1-\epsilon)\FOPT(V,1)-\delta,    
\]
where the first inequality is from Lemma~\ref{lem:fractional-optimal-solution}, the second inequality is from Lemma~\ref{lem:frac_efficient}, and the last inequality is from $W\geq 1-\epsilon$.
\end{proof}

We turn to analyzing the average sensitivity of Algorithm~\ref{alg:smallknapsack}.
We start by proving the following auxiliary lemma.
\begin{lemma}\label{lem:smallsum}
For any $j\in V$, it holds that
\[
    \int_{1-\epsilon}^{1}\sum_{i\in V} 1_{j\in S(W)\triangle S^i(W)}\mathrm{d}W\leq 1+\epsilon.
\]
\end{lemma}
\begin{proof}
Note that, $j\in S(W)$ implies $j\in S^i(W)$ unless $i=j$. Therefore, we have
\begin{align}
    \sum_{i\in V} 1_{j\in S(W)\triangle S^i(W)}
    \leq \sum_{i\in V, i\neq j}1_{j\in S^i(W)\setminus S(W)} + 1.\label{eq:removeij}
\end{align}

Suppose $i\neq j$. Then, $j\in S^i(W)\setminus S(W)$ holds if and only if $i\in S(W)$ and $w(\{1,\dots, j\}\setminus \{i\}) \leq W < w(\{1,\dots, j\})$. Let $k$ be the largest index such that $w(\{1,\dots, k\})\leq 1$. Then, we have
\begin{align*}
    &\int_{1-\epsilon}^{1}\sum_{i\in V}1_{j\in S(W)\triangle S^i(W)}\mathrm{d}W
    \leq \int_{1-\epsilon}^{1}\left(\sum_{i\in V, i\neq j}1_{j\in S^i(W)\setminus S(W)} + 1\right)\mathrm{d}W\\
    &\leq \int_{1-\epsilon}^{1}\sum_{i = 1}^k 1_{w(\{1,\dots, j\}\setminus \{i\}) \leq W < w(\{1,\dots, j\})}\mathrm{d}W + \epsilon
    = \sum_{i = 1}^k\int_{1-\epsilon}^{1}1_{w(\{1,\dots, j\}\setminus \{i\}) \leq W < w(\{1,\dots, j\})}\mathrm{d}W + \epsilon\\
    &\leq \sum_{i = 1}^k w(i) + \epsilon
     \leq 1 + \epsilon.
\end{align*}
Here, the first inequality is given by~\eqref{eq:removeij}, the second inequality is due to the rephrase of the condition $j\in S^i(W)\setminus S(W)$ mentioned above, the third inequality is by $w(\{1,\dots, j\})-w(\{1,\dots, j\}\setminus \{i\}) \leq  w(i)$ and the distribution of $W$, and the last inequality arises from the definition of $k$.
\end{proof}

Let $\mathcal{A}_{\text{SMALL}}$ (resp., $\mathcal{A}_{\text{SMALL}}^i\;(i \in V)$) be the output distribution of Algorithm~\ref{alg:smallknapsack} run on the instance $V$ (resp., $V\setminus \{i\}$). Now, we show the following:
\begin{lemma}\label{lem:smallsensitivity}
Algorithm~\ref{alg:smallknapsack} has average sensitivity $\epsilon^{-1}+1$.
\end{lemma}
\begin{proof}
For each $i\in V$, we bound the earth mover's distance between $\mathcal{A}_{\text{SMALL}}$ to $\mathcal{A}_{\text{SMALL}}^i$ by transporting the maximum possible probability mass of $S(W)$ to $S^i(W)$. Now, we have
\begin{align*}
    \frac{1}{n}\sum_{i\in V}\EM(\mathcal{A}_{\text{SMALL}},\mathcal{A}_{\text{SMALL}}^i)
    &\leq \frac{1}{n}\sum_{i\in V}\left(\frac{1}{\epsilon}\int_{1-\epsilon}^{1}\left(\sum_{j\in V} 1_{j\in S(W)\triangle S^i(W)}\right)\mathrm{d}W\right)\\
    &\leq \frac{1}{n\epsilon}\sum_{j\in V}\left(1+\epsilon\right)
    = \epsilon^{-1}+1,
\end{align*}
where the first inequality arises from the transport of the probability mass and the definition of $W$ in Algorithm~\ref{alg:smallknapsack}, and the second inequality is obtained by swapping the sum and integral and applying Lemma~\ref{lem:smallsum}.
\end{proof}
Theorem~\ref{thm:small} is immediate by Lemmas~\ref{lem:smallapprox} and~\ref{lem:smallsensitivity}.



\section{General Case}\label{sec:general}
In this section, we consider the general case of the knapsack problem, and prove the following:
\begin{theorem}\label{thm:main}
For any $\epsilon>0$, there exists a $(1-\epsilon)$-approximation algorithm for the knapsack problem with average sensitivity $O(\epsilon^{-1}\log \epsilon^{-1})$.
\end{theorem}
Note that this algorithm takes exponential time.
We will discuss on improving time complexity in~\cref{sec:fptas}.

\subsection{Technical Overview and Algorithm Description}\label{sec:overview}
First, we explain the intuition behind our algorithm (Algorithm~\ref{alg:knapsack}).
As in Section~\ref{sec:small-weight}, we assume that the weight limit is $1$ without loss of generality.
We fix a parameter $0 < \epsilon < 0.05$.
We say that an item is \emph{large} if its value is at least $\epsilon \cdot \FOPT(V)$ and \emph{small} otherwise.
If all items are small, \Call{ModifiedGreedy}{} (Algorithm~\ref{alg:smallknapsack}) with the parameter $\epsilon$ has an approximation ratio $1-O(\epsilon)$ and average sensitivity $O(\epsilon^{-1})$, and we are done.
On the contrary, if all items are large, the procedure of outputting the optimal solution has average sensitivity $O(\epsilon^{-1})$ because any feasible solution has cardinality at most $\epsilon^{-1}$.
We combine these two observations to obtain an algorithm for the general case as follows.
Note that here we don't concern about the running time; we accelerate it in~\cref{sec:fptas}.

Let $L \subseteq V$ be the set of large items.
First, we take a subset $R \subseteq L$ with $w(R) \leq 1$ that maximizes $v(R)+\OPT(V\setminus L, 1-w(R))$. (Note that in this section, we do not take efficiency into consideration and therefore we assume that $\OPT(V\setminus L, 1-w(R))$ can be computed.)
Then, we output $R\cup \Call{ModifiedGreedy}{\epsilon, V\setminus L, 1-w(R)}$, where $\Call{ModifiedGreedy}{\epsilon, U, W}$ runs $\Call{ModifiedGreedy}{\epsilon, U}$ after replacing $w(i)$ with $w(i)/W$ for each $i \in U$.

However, this algorithm has the following two issues.
\begin{enumerate}
\item If the value of $\FOPT(V)$ changes upon deleting an item, the set $L$ may drastically change. For example, consider the case when there are many items with values slightly less than $\epsilon \cdot \FOPT(V)$. Then, the set of small items added to the output may drastically change.
\item Even when the value $\FOPT(V)$ does not change, if the choice of $R$ changes upon deleting an item, the value of $1-w(R)$ also changes. Thus, the set of small items added to the output may drastically change.
\end{enumerate}

To resolve the first issue, we sample a value threshold $c=O(\epsilon\cdot \FOPT(V))$ that classifies items as large or small from an appropriate distribution, instead of fixing it to $\epsilon\cdot \FOPT(V)$.

Now, we consider the second issue.
Suppose we have sampled the same value threshold $c$ both before and after deleting an item $i \in V$.
There are two cases in which $R$ changes after deleting the item $i$.
\begin{enumerate}
\item If the item $i$ is large and $i \in R$, then the algorithm should change the choice of $R$ because the item $i$ would no longer exist.
\item If the item $i$ is small, then the algorithm may change the choice of $R$ because the value $\FOPT(V\setminus L, 1-w(R))$ (for the original $R$) may decrease.
\end{enumerate}
The first case is easy to resolve; $R$ contains $O(\epsilon^{-1})$ many items and therefore, by taking average over $i$, this case contributes to the average sensitivity by $O(\epsilon^{-1})$.

To address the second case (deleting small items), we ensure that the distribution of $R$ does not change significantly with small decreases in $\FOPT(V\setminus L, 1-w(R))$.
To this end, instead of considering the $R$ with the maximum value of $v(R)+\FOPT(V\setminus L, 1-w(R))$ among $R\subseteq L$, we apply the exponential mechanism~\cite{mcsherry2007mechanism}.
Specifically, we sample $R$ with probability proportional to the exponential of this value with appropriate scaling and rounding (see Line~\ref{line:exponential} in Algorithm~\ref{alg:knapsack}).
To ensure that the mechanism outputs a $(1-\epsilon)$-approximate solution, we reduce the number of candidates that can possibly be $R$ to a constant.
This is implemented by taking only one candidate $A_t$ from the family of sets $R$ with $w(R)\in \left[tc,(t+1)c\right)$ for each integer $t$ (see Line~\ref{line:smallestweight} in Algorithm~\ref{alg:knapsack}).
For technical reasons, we apply an exponential mechanism for the value $tc+\FOPT(V\setminus L, 1-w(A_t))$, rather than the exact value $v(A_t)+\OPT(V\setminus L, 1-w(A_t))$ (see Line~\ref{line:defx}).
In Sections~\ref{subsec:approximation-ratio} and~\ref{subsec:sensitivity}, we analyze the approximation ratio and average sensitivity of Algorithm~\ref{alg:knapsack}, respectively.

\begin{algorithm}[t!]
\caption{A $(1-O(\epsilon))$-approximation algorithm with small average sensitivity}\label{alg:knapsack}
\Procedure{\emph{\Call{StableOnAverageKnapsack}{$\epsilon,V$}}}{
Sample $c$ from the uniform distribution over $[\epsilon\cdot \FOPT(V),2\epsilon\cdot \FOPT(V)]$\;
Let $L\subseteq V$ be the set of items with value at least $c$\;
Let $l=\lfloor\FOPT(V)/c\rfloor$\;
\For{$t=0,\dots, l$}{
    \If{there is a subset of $L$ with value in $\left[tc,(t+1)c\right)$}{\label{line:condx}
        Let $A_{t}\subseteq L$ be the set of items with smallest weight that satisfies $tc \leq v(A_{t}) < (t+1)c$\label{line:smallestweight}\; (if there are multiple choices, choose the lexicographically smallest one)\;
        Let $x_{t}=tc + \FOPT(V\setminus L, 1-w(A_{t}))$\;\label{line:defx}
    }\Else{
        Let $A_{t}=\emptyset$ and $x_{t}=-\infty$\;
    }
}
Let $d=\frac{c}{10\log(\epsilon^{-1})}=O\left(\frac{\epsilon}{\log (\epsilon^{-1})}\FOPT(V)\right)$\;\label{line:choicec}
Sample $t^{\circ} \in \{0,\ldots,l\}$ with probability proportional to $\exp(x_{t^{\circ}}/d)$ and let $R=A_{t^{\circ}}$\;\label{line:exponential}
\Return $R\cup \Call{ModifiedGreedy}{V\setminus L, 1-w(A_{t})}$\;\label{line:small}
}
\end{algorithm}

\subsection{Approximation Ratio}\label{subsec:approximation-ratio}

First, we analyze the approximation ratio of Algorithm~\ref{alg:knapsack}.
Let $S^*$ be a set of items that attains $\OPT(V)$. Let $t^* \geq 0$ be an integer such that $v(S^*\cap L)\in \left[t^*c, (t^*+1)c\right)$.

The following lemma bounds the loss in the approximation ratio caused by considering only $A_{0},\dots, A_{l}$ instead of all subsets of $L$ as candidate sets that can possibly be $R$.

\begin{lemma}\label{lem:errort}
    $x_{t^*}\geq \left(1-4\epsilon \right)\OPT(V)$ holds.
\end{lemma}
\begin{proof}
We have
\begin{align*}
    x_{t^*}&=t^*c+\FOPT(V\setminus L, 1-w(A_{t^*}))
    \geq t^*c+\FOPT(V\setminus L, 1-w(S^*\cap L))\\
    &\geq t^*c+v(S^*\setminus L)
    \geq v(S^*\cap L)-c+v(S^*\setminus L)
    = \OPT(V)-c
    \geq (1-4\epsilon)\OPT(V),
\end{align*}
where the first equality is from the definition of $x_{t^*}$, the first inequality is from $w(A_{t^*})\leq w(S^*\cap L)$, which is ensured by Line~\ref{line:smallestweight} in the algorithm, the second inequality is from $\FOPT(V\setminus L, 1-w(S^*\cap L))\geq \OPT(V\setminus L, 1-w(S^*\cap L))=v(S^*\setminus L)$, the third inequality is from the definition of $t^*$, the second equality is from the optimality of $S^*$, and the last inequality is from $c\leq 2\epsilon\cdot \FOPT(V)\leq 4\epsilon\cdot \OPT(V)$.
\end{proof}

Next we analyze the loss in the approximation ratio caused by the exponential method applied at Line~\ref{line:exponential}.
We prove the following.
\begin{lemma}\label{lem:errorexp}
$\E[x_{t^{\circ}}]\geq (1-3\epsilon) x_{t^*}$ holds.
\end{lemma}
\begin{proof}
We have
\begin{align}
    &\Pr[x_{t^{\circ}}\leq (1-\epsilon)x_{t^*}]
    = \frac{\sum_{t\in \{0,\dots, l\}: x_{t}\leq (1-\epsilon)x_{t^*}}\exp\left(\frac{x_{t}}{d}\right)}{\sum_{t\in \{0,\dots, l\}}\exp\left(\frac{x_{t}}{d}\right)}
    \leq \frac{(l+1) \exp\left(\frac{(1-\epsilon)x_{t^*}}{d}\right)}{\exp\left(\frac{x_{t^*}}{d}\right)} \nonumber \\
    &= (l+1) \exp\left(- \frac{\epsilon x_{t^*}}{d}\right) \nonumber \\
    &\leq (l+1) \exp\left(- \frac{\epsilon(1-4\epsilon)\OPT(V)}{d}\right)
    \leq (l+1) \exp\left(-\frac{5}{2}\log \epsilon^{-1}(1-4\epsilon)\right) \nonumber \\
    &= (l+1) \epsilon^{\frac{5}{2} (1-4\epsilon)}
    \leq 2\epsilon^{-1}\cdot \epsilon^2 = 2\epsilon.\label{eq:2c2c}
\end{align}
Here, the first equality is from Line~\ref{line:exponential} of the algorithm, the first inequality is from the fact that there are at most $l+1$ indices $t$ with $x_{t}\leq (1-\epsilon)x_{t^*}$, the second inequality is from Lemma~\ref{lem:errort}, the third inequality is from
\[
    d= \frac{c}{10\log \epsilon^{-1}} \leq \frac{\epsilon}{5\log \epsilon^{-1}}\cdot \FOPT(V)\leq \frac{5}{2}\cdot \frac{\epsilon}{\log \epsilon^{-1}}\cdot \OPT(V),
\]
and the last inequality is from $l+1\leq \epsilon^{-1}+1\leq 2\epsilon^{-1}$ and $\frac{5}{2}(1-4\epsilon)\geq 2$, which is from $\epsilon\leq 0.05$.
Therefore, we have
\[
    \E[x_{t^{\circ}}]
    \geq (1-\epsilon)x_{t^*}\Pr[x_{t^{\circ}} > (1-\epsilon)x_{t^*}]
    \geq (1-\epsilon) \left(1-2\epsilon\right)x_{t^*} \geq (1-3\epsilon)x_{t^*}.
\]
Here, the second inequality is from~\eqref{eq:2c2c}.
\end{proof}

Combining the lemmas above yields the following.
\begin{lemma}\label{lem:stepapprox}
$\E[v(A_{t^{\circ}})+\FOPT(V\setminus L, 1-w(A_{t^{\circ}}))]\geq (1-7\epsilon) \OPT(V)$ holds.
\end{lemma}
\begin{proof}
We have
\begin{align*}
    &\E[v(A_{t^{\circ}})+\FOPT(V\setminus L, 1-w(A_{t^{\circ}}))]\\
    &\geq \E[x_{t^{\circ}}]
    \geq (1-3\epsilon)x_{t^*}
    \geq (1-3\epsilon)(1-4\epsilon)\OPT(V)
    \geq (1-7\epsilon)\OPT(V),
\end{align*}
where the first inequality is from Line~\ref{line:defx} of the algorithm, the second inequality is from Lemma~\ref{lem:errorexp}, and the third inequality is from Lemma~\ref{lem:errort}.
\end{proof}

Now we bound the approximation ratio of Algorithm~\ref{alg:knapsack}. Let $\mathcal{A}_{\mathrm{SMALL}}(V',W')$ be the output of Algorithm~\ref{alg:smallknapsack} on $V'$, where all weights in the input are divided by $W'$.
\begin{lemma}\label{lem:largeapprox}
The approximation ratio of Algorithm~\ref{alg:knapsack} is at least $1-12\epsilon$.
\end{lemma}
\begin{proof}
Let $\mathcal{A}$ be the output distribution of Algorithm~\ref{alg:knapsack} applied on the instance $(V,1)$.
We have
\begin{align*}
    \E[\mathcal{A}]&=\E[v(A_{t^{\circ}})+v\left(\mathcal{A}_{\mathrm{SMALL}}(V\setminus L, 1-w(A_{t^{\circ}}))\right)]\\
    &\geq \E[v(A_{t^{\circ}})+(1-\epsilon)\FOPT(V\setminus L, 1-w(A_{t^{\circ}}))-c]\\
    &\geq \E[(1-\epsilon)(v(A_{t^{\circ}})+\FOPT(V\setminus L, 1-w(A_{t^{\circ}})))-c]\\
    &\geq \E[(1-\epsilon)(1-7\epsilon)\OPT(V)-c]
    \geq (1-12\epsilon)\OPT(V),
\end{align*}
where the first inequality is from Lemma~\ref{lem:smallapprox}, the third inequality is from Lemma~\ref{lem:stepapprox}, and the last inequality is from $c\leq 2\FOPT(V)\leq 4\OPT(V)$.
\end{proof}

\subsection{Average Sensitivity}\label{subsec:sensitivity}

In this section, we discuss bounding the average sensitivity of Algorithm~\ref{alg:knapsack}.
For the parameters used in the algorithm applied on the instance $(V \setminus \{i\},1)$, we use symbols $c^i, d^i$, and $R^i$ to denote $c$, $d$, and $R$, respectively (and use $c$, $d$, and $R$ for the instance $(V,1)$).
Similarly, we use the symbols $A^i_{t}$ and $x^i_{t}$ for $t=0,\dots, l$ to denote $A_{t}$ and $x_{t}$, respectively, for the instance $(V \setminus \{i\},1)$.
We first prove that the distributions of $R$ and $R^i$ are sufficiently close on average, where the average is taken over $i\in V$ (Lemma~\ref{lem:combined}). Subsequently, we combine the analysis for large and small items (Lemma~\ref{lem:largesensitivity}).

The distributions of $R$ and $R^i$ may differ for the following two reasons: the difference between the distributions of $c$ and $c^i$, and the absence of the item $i$ in the instance $V\setminus \{i\}$.
We first forcus on the second reason.
Specifically, we fix the value $\hat{c}$ and assume $c=c^i=\hat{c}$.
In this case, we prove that the distribution of $R$ is sufficiently close to that of $R^i$ on average, where the average is taken over $i\in V$.
To this end, we analyze the following quantity:
\begin{align}
    \sum_{i\in V}\sum_{A\subseteq V}\max\left(0,\Pr[R=A \mid c=\hat{c}]-\Pr[R^i=A \mid c^i=\hat{c}]\right).\label{eq:befdiffs}
\end{align}
Because $\Pr[R=A \mid c=\hat{c}]$ is positive only if $A=A_{t}$ for some $t\in \{0,\dots, l\}$, we have
\begin{align}
    \eqref{eq:befdiffs} = \sum_{i\in V}\sum_{t\in \{0,\dots, l\}}\max\left(0,\Pr[R=A_{t} \mid c=\hat{c}]-\Pr[R^i=A_{t} \mid c^i=\hat{c}]\right).\label{eq:diffs}
\end{align}
We analyze the sum in~\eqref{eq:diffs} by dividing it into two cases: $i\in L$ and $i\in V\setminus L$.
We first show that $x_t\geq x^i_t$ holds for all $i$ and $t$.
\begin{lemma}\label{lem:decreasex}
For any $t\in \{0,\dots, l\}$ and $i\in V$, we have $x_t\geq x^i_t$.
\end{lemma}
\begin{proof}
If $x_t=-\infty$, then we have $x^i_t=-\infty$ due to Line~\ref{line:condx} of Algorithm~\ref{alg:knapsack}.
Hereafter, we assume $x_t\neq -\infty$. We prove the lemma by considering the following three cases:
\begin{claim}\label{cla:1}
    If $i\in L\setminus A_t$, we have $x_{t} = x^i_{t}$.
\end{claim}
\begin{proof}
    Since $i\not \in A_{t}$, we have $A_{t}=A^i_{t}$. Therefore, we have
    \begin{align*}
        x_{t}&=tc+\FOPT(V\setminus L,1-w(A_{t})) \\
        &=tc+\FOPT((V\setminus \{i\})\setminus (L\setminus \{i\}),1-w(A^i_{t}))=x^i_{t}, 
    \end{align*}
    where the second equality is from $i\in L$ and $A_{t}=A^i_{t}$.
\end{proof}
\begin{claim}\label{cla:2}
    If $i\in A_t$, we have $x_{t} \geq x^i_{t}$.
\end{claim}
\begin{proof}
    We have
    \begin{align*}
        x_{t}&=tc+\FOPT(V\setminus L,1-w(A_{t}))\\
        &\geq tc+\FOPT((V\setminus \{i\})\setminus (L\setminus \{i\}),1-w(A^i_{t}))= x^i_{t}, 
    \end{align*}
    where the inequality is from $i\in L$ and $w(A_{t})\leq w(A^i_{t})$, which follows from Line~\ref{line:smallestweight} of Algorithm~\ref{alg:knapsack}.
\end{proof}
\begin{claim}\label{cla:3}
    If $i\in V\setminus L$, we have $x_{t} \geq x^i_{t}$.
\end{claim}
\begin{proof}
    We have
    \begin{align}
        x_{t}&=tc+\OPT(V\setminus L, 1-w(A_{t})) \nonumber \\
        &= tc+\OPT(V\setminus L, 1-w(A^i_{t})) \nonumber \\
        &\geq tc+\OPT((V\setminus \{i\})\setminus (L\setminus \{i\}), 1-w(A^i_{t}))=x^i_{t},
    \end{align}
    where the second equality is because $L=L^i$ and hence $A_{t}=A^i_{t}$, and the last equality is from $i\not \in L$.
\end{proof}
Then we complete the case analysis and the lemma is proved.
\end{proof}

The next lemma handles the case $i\in L$. 
\begin{lemma}\label{lem:baai2}
For any $\hat{c} \in \mathbb{R}$, we have
\[
    \sum_{i\in L}\sum_{t\in \{0,\dots, l\}}\max\left(0,\Pr[R=A_{t}\mid c=\hat{c}]-\Pr[R^i=A_{t}\mid c^i=\hat{c}]\right)\leq \epsilon^{-1}.
\]
\end{lemma}
\begin{proof}
We have
\begin{align}
    &\sum_{i\in L}\sum_{t\in \{0,\dots, l\}}\max\left(0,\Pr[R=A_{t}\mid c=\hat{c}]-\Pr[R^i=A_{t}\mid c^i=\hat{c}]\right) \nonumber \\
    &= \sum_{i\in L}\sum_{t: i\not \in A_{t}}\max\left(0,\Pr[R=A_{t}\mid c=\hat{c}]-\Pr[R^i=A_{t}\mid c^i=\hat{c}]\right)\label{eq:notin}\\
    &+ \sum_{i\in L}\sum_{t: i\in A_{t}}\max\left(0,\Pr[R=A_{t}\mid c=\hat{c}]-\Pr[R^i=A_{t}\mid c^i=\hat{c}]\right)\label{eq:in}.
\end{align}

We first consider the first term~\eqref{eq:notin}.
Now, for $t\in \{0,\dots, l\}$ and $i\not \in A_{t}$, we have
\begin{align*}
    \Pr[R=A_{t}\mid c=\hat{c}] &= \frac{\exp(x_{t}/\hat{d})}{\sum_{t'\in \{0,\dots, l\}}\exp(x_{t'}/\hat{d})}\\
    &= \frac{\exp(x^i_{t}/\hat{d})}{\sum_{t'\in \{0,\dots, l\}}\exp(x_{t'}/\hat{d})}\\
    &\leq \frac{\exp(x^i_{t}/\hat{d})}{\sum_{t'\in \{0,\dots, l\}}\exp(x^i_{t'}/\hat{d})}=\Pr[R^i=A_{t}\mid c^i=\hat{c}].
\end{align*}
Here, the first and the last equality is from Line~\ref{line:exponential} of Algorithm~\ref{alg:knapsack}, the second equality is from Claim~\ref{cla:1}, and the inequality is from Lemma~\ref{lem:decreasex}.
Therefore the value of~\eqref{eq:notin} is zero.

Now, we evaluate the second term~\eqref{eq:in}. We have
\begin{align*}
    \eqref{eq:in}
    &\leq \sum_{i\in L}\sum_{t:i\in A_{t}}\Pr[R=A_{t}\mid c=\hat{c}]\\
    &= \sum_{t\in \{0,\dots, l\}}|A_{t}|\cdot \Pr[R=A_{t}\mid c=\hat{c}]\\
    &\leq \epsilon^{-1}\sum_{t\in \{0,\dots, l\}}\Pr[R=A_{t}\mid c=\hat{c}] \leq \epsilon^{-1},
\end{align*}
where the second inequality is from $|A_t|\leq \FOPT(V)/c\leq \epsilon^{-1}$.
Therefore the lemma is proved.
\end{proof}

Now we consider the case $i\in V\setminus L$. 
\begin{lemma}\label{lem:baai3}
For any $\hat{c} \in \mathbb{R}$, we have
\[
    \sum_{i\in V\setminus L}\sum_{t\in \{0,\dots, l\}}\max\left(0,\Pr[R=A_{t}\mid c=\hat{c}]-\Pr[R^i=A_{t}\mid c^i=\hat{c}]\right)\leq 10\epsilon^{-1} \log \epsilon^{-1}.
\]
\end{lemma}
\begin{proof}
Let $\hat{d}=\frac{\hat{c}}{10\log\left(\epsilon^{-1}\right)}$. Then, we have
\begin{align}
    &\Pr[R=A_{t}\mid c=\hat{c}]-\Pr[R^i=A_{t}\mid c^i=\hat{c}] \nonumber \\
    &= \frac{\exp(x_{t}/\hat{d})}{\sum_{t\in \{0,\dots, l\}}\exp(x_{t}/\hat{d})}- \frac{\exp(x^i_{t}/\hat{d})}{\sum_{t\in \{0,\dots, l\}}\exp(x^i_{t}/\hat{d})}
    \leq \frac{\exp(x_{t}/\hat{d})-\exp(x^i_{t}/\hat{d})}{\sum_{t\in \{0,\dots, l\}}\exp(x_{t}/\hat{d})} \nonumber \\
    &= \left(1-\exp\left(-\frac{x_{t}-x^i_{t}}{\hat{d}}\right)\right)\frac{\exp(x_{t}/\hat{d})}{\sum_{t\in \{0,\dots, l\}}\exp(x_{t}/\hat{d})}
    \leq \frac{x_{t}-x^i_{t}}{\hat{d}}\Pr[R=A_{t}\mid c=\hat{c}].\label{eq:explinear}
\end{align}
Here, the first equality is from Line~\ref{line:exponential} of Algorithm~\ref{alg:knapsack}, the first inequality is from Lemma~\ref{lem:decreasex}, and the last inequality is from $1-\exp(-x)\leq x$.
Now, we have
\begin{align*}
    &\sum_{i\in V\setminus L}\sum_{t\in \{0,\dots, l\}}\max\left(0,\Pr[R=A_{t}\mid c=\hat{c}]-\Pr[R^i=A_{t}\mid c^i=\hat{c}]\right)\\
    &\leq \sum_{i\in V\setminus L}\sum_{t\in \{0,\dots, l\}}\frac{x_{t}-x^i_{t}}{\hat{d}}\cdot \Pr[R=A_{t}\mid c=\hat{c}]\\
    &\leq \sum_{t\in \{0,\dots, l\}}\frac{\FOPT(V\setminus L,1-w(A_{t}))}{\hat{d}}\cdot \Pr[R=A_{t}\mid c=\hat{c}]\\
    &\leq \frac{\FOPT(V)}{\hat{d}}\sum_{t\in \{0,\dots, l\}}\Pr[R=A_{t}\mid c=\hat{c}]
    \leq 10\epsilon^{-1} \log \epsilon^{-1}.
\end{align*}
Here, the first inequality is obtained from~\eqref{eq:explinear} and Claim~\ref{cla:3}. The second inequality is obtained from:
\begin{align*}
    \sum_{i\in V\setminus L}(x_{t}-x^i_{t})
    &=\sum_{i\in V\setminus L}(\FOPT(V\setminus L, 1-w(A_{t}))-\FOPT((V\setminus L)\setminus \{i\},1-w(A_{t})))\\
    &\leq \FOPT(V\setminus L,1-w(A_{t})),
\end{align*}
which is obtained from Lemma~\ref{lem:fracdiffsum}. The last inequality is from Line~\ref{line:choicec} of Algorithm~\ref{alg:knapsack}.
\end{proof}

Combining Lemmas~\ref{lem:baai2} and~\ref{lem:baai3}, we obtain the following.
\begin{lemma}\label{lem:combined}
We have
\[
    \sum_{i\in V}\sum_{A\subseteq V}\max \left(0, \Pr[R=A\mid c=\hat{c}]-\Pr[R^i=A\mid c^i=\hat{c}]\right)\leq 11\epsilon^{-1}\log \epsilon^{-1}.
\]
\end{lemma}
\begin{proof}
The claim immediately follows from Lemmas~\ref{lem:baai2} and~\ref{lem:baai3} and the fact that $V$ is a disjoint union of $V\setminus L$ and $L$.
\end{proof}

Now, we evaluate the average sensitivity.
Let $\mathcal{A}^i$ be the output distribution of Algorithm~\ref{alg:knapsack} applied on the instance $V\setminus \{i\}$.

To bound the earth mover's distance, we consider transporting the probability mass in such a way the mass of $\mathcal{A}$ corresponding to a particular choice of $c$ is transported to that of $\mathcal{A}^i$ for the same $c^i$ (as far as we can).
For the same choice of $c$ and $c^i$, we consider transporting the probability mass in such a way that the mass corresponding to a particular choice of $R$ is transported to that for the same $R^i$ (as far as we can).
For the same choice of $c,R$ and $c^i,R^i$, we consider transporting the probability mass in a manner similar to the analysis of Algorithm~\ref{alg:smallknapsack}.
The remaining mass is transported arbitrarily.

First, we bound the contribution of the difference between the distributions of $c$ and $c^i$ to the earth mover's distance.
Let $f$ and $f^i$ be the probability density functions of $c$ and $c^i$, respectively. Precisely, $f(c)=\frac{1}{\epsilon\cdot \FOPT(V)}$ if $\epsilon\cdot \FOPT(V)\leq c\leq 2 \epsilon\cdot \FOPT(V)$ and $0$ otherwise. Similarly, $f^i(c^i)=\frac{1}{\epsilon\cdot \FOPT(V\setminus \{i\})}$ if $\epsilon\cdot \FOPT(V\setminus \{i\})\leq c^i\leq 2 \epsilon\cdot \FOPT(V\setminus \{i\})$ and $0$ otherwise.

\begin{lemma}\label{lem:errorbyc}
We have
\[
    \sum_{i\in V}\int_{\epsilon\cdot \FOPT(V)}^{2\epsilon\cdot \FOPT(V)}\max\left(0,f(\hat{c})-f^i(\hat{c})\right)\mathrm{d}\hat{c}\leq 2.
\]
\end{lemma}
\begin{proof}
We have
\begin{align*}
    &\sum_{i\in V}\int_{\epsilon\cdot \FOPT(V)}^{2\epsilon\cdot \FOPT(V)}\max\left(0,f(\hat{c})-f^i(\hat{c})\right)\mathrm{d}\hat{c} 
    = \sum_{i\in V}\int_{\max\{\epsilon\cdot \FOPT(V),2\epsilon\cdot \FOPT(V\setminus \{i\})\}}^{2\epsilon\cdot \FOPT(V)}\frac{1}{\epsilon\cdot \FOPT(V)}\mathrm{d}\hat{c} \\
    &\leq \frac{1}{\epsilon\cdot \FOPT(V)}\sum_{i\in V}\left(2\epsilon\cdot (\FOPT(V)-\FOPT(V\setminus \{i\}))\right)
    \leq \frac{1}{\epsilon\cdot \FOPT(V)}\cdot 2\epsilon\cdot \FOPT(V) = 2, 
\end{align*}
where the first inequality is from the fact that $f(\hat{c})\leq f^i(\hat{c})$ holds if $\epsilon\cdot \FOPT(V)\leq \hat{c}\leq 2\epsilon\cdot \FOPT(V\setminus \{i\})$, and the last inequality is from Lemma~\ref{lem:fracdiffsum}.
\end{proof}

Finally, we bound the average sensitivity.

\begin{lemma}\label{lem:largesensitivity}
The average sensitivity of Algorithm~\ref{alg:knapsack} is at most $12\epsilon^{-1} \log \epsilon^{-1}$.
\end{lemma}
\begin{proof}
We have 
\begin{align*}
    &\frac{1}{n}\sum_{i\in V}\EM(\mathcal{A},\mathcal{A}^i)\\
    &\leq \frac{1}{n}\sum_{i\in V}\int_{\epsilon\cdot \FOPT(V)}^{2\epsilon\cdot \FOPT(V)}\max\left(0,f(\hat{c})-f^i(\hat{c})\right)\mathrm{d}\hat{c}\cdot n\nonumber \\
    & \quad +\frac{1}{n}\sum_{i\in V}\int_{\epsilon\cdot \FOPT(V)}^{2\epsilon\cdot \FOPT(V)}\sum_{A\subseteq V}{\Bigl(}\max \left(0, \Pr[R=A\mid c=\hat{c}]-\Pr[R^i=A\mid c^i=\hat{c}]\right)\cdot n\nonumber \\
    &\quad +\Pr[R=A \mid c=\hat{c}]\cdot \EM\left(\mathcal{A}_{\mathrm{SMALL}}(V\setminus L,1-w(A)),\mathcal{A}_{\mathrm{SMALL}}((V\setminus L)\setminus \{i\},1-w(A))\right){\Bigr)}f(\hat{c})\mathrm{d}\hat{c}\\
    &\leq 2+ \frac{1}{n}\int_{\epsilon\cdot \FOPT(V)}^{2\epsilon\cdot \FOPT(V)}\left(11\epsilon^{-1} \log \epsilon^{-1}n +\sum_{i\in V}\sum_{A\subseteq V}\Pr[R=A \mid c=\hat{c}](\epsilon^{-1}+1)\right)f(\hat{c})\mathrm{d}\hat{c}\\
    &= 2+ \frac{1}{n}\int_{\epsilon\cdot \FOPT(V)}^{2\epsilon\cdot \FOPT(V)}\left(11\epsilon^{-1} \log \epsilon^{-1}n +(\epsilon^{-1}+1)\cdot n\right)f(\hat{c})\mathrm{d}\hat{c}\\
    &=2 + 11\epsilon^{-1} \log \epsilon^{-1}+\epsilon^{-1}+1\leq 12\epsilon^{-1} \log \epsilon^{-1},
\end{align*}
where the first inequality is due to the transport of the probability mass, and the second inequality is from Lemmas~\ref{lem:errorbyc},~\ref{lem:combined}, and~\ref{lem:smallsensitivity}.
\end{proof}

\begin{proof}[Proof of Theorem~\ref{thm:main}]
Applying Lemma~\ref{lem:largeapprox} and Lemma~\ref{lem:largesensitivity} and replacing $\epsilon$ with $\min(0.05,\epsilon/12)$ proves this theorem.
\end{proof}


\section{FPRAS}\label{sec:fptas}

In this section, we describe a modification of Algorithm~\ref{alg:knapsack} so that it runs in polynomial time in $n$ and $\epsilon^{-1}$, and prove Theorem~\ref{thm:intro}.

As the fractional knapsack problem can be solved in polynomial time, the only issue is to compute $A_t$ at Line~\ref{line:smallestweight} of Algorithm~\ref{alg:knapsack}.
To address this issue, we round the values of the items so that $A_t$ can be computed efficiently.
We use the same rounding scheme as the textbook algorithm for the knapsack problem~\cite{ibarra1975fast,korte2011combinatorial}. Without loss of generality, we assume that all the items have weight at most $1$.

\begin{algorithm}[t!]
\caption{An FPRAS with small average sensitivity}\label{alg:fptas}
\Procedure{\emph{\Call{FPRAS}{$\epsilon,V$}}}{
    Sample $\delta$ from the uniform distribution over $\left[\frac{\FOPT(V)}{n}\epsilon,2\frac{\FOPT(V)}{n}\epsilon\right]$\;\label{line:choosef}
    Let $V'$ be an empty set of items.\;
    \For{$i\in V$}{
        $v'(i)\leftarrow \left\lfloor \frac{v(i)}{\delta} \right\rfloor$\;\label{line:roundd}
        Add an item $i$ with value $v'(i)$ and weight $w(i)$ to $V'$.
    }
    \Return \Call{StableOnAverageKnapsack}{$\epsilon,V',W'$}, where $W'=\left\lfloor \frac{1}{\delta} \right\rfloor$.
}
\end{algorithm}

Our algorithm is presented in Algorithm~\ref{alg:fptas}.
After the rounding, the values of the items and the weight limit are integers at most $n\epsilon^{-1}$.
Hence, we can compute $A_t$ at Line~\ref{line:smallestweight} of Algorithm~\ref{alg:knapsack} in polynomial time in $n$ and $\epsilon^{-1}$ using a standard dynamic programming~\cite{ibarra1975fast,korte2011combinatorial}.

Let $\hat{\mathcal{A}}$ be the output distribution of Algorithm~\ref{alg:fptas}. We have the following lemma.
\begin{lemma}\label{lem:fpapprox}
Let $1-\epsilon'$ be the approximation ratio of Algorithm~\ref{alg:knapsack}. Then, Algorithm~\ref{alg:fptas} has approximation ratio $1-4\epsilon-\epsilon'$.
\end{lemma}
\begin{proof}
Let $S^*$ be the optimal solution for the instance $(V,1)$. Then, we have
\begin{align}
    \OPT(V',W')&\geq v'(S^*)
    = \sum_{i\in S^*}\left\lfloor \frac{v(i)}{\delta} \right\rfloor
    \geq \sum_{i\in S^*}\left(\frac{v(i)}{\delta}-1\right)
    \geq \frac{\OPT(V)}{\delta}-n \nonumber \\
    &= \left(1-\frac{n\delta}{\OPT(V)}\right)\frac{\OPT(V)}{\delta}
    \geq \left(1-\frac{2\FOPT(V)\epsilon}{\OPT(V)}\right)\frac{\OPT(V)}{\delta} \nonumber \\
    & \geq \left(1-4\epsilon\right)\frac{\OPT(V)}{\delta},\label{eq:m2e}
\end{align}
where the third inequality is from $|S^*|\leq n$, the fourth inequality is from Line~\ref{line:choosef} of Algorithm~\ref{alg:fptas}, and the last inequality is from $\FOPT(V)\leq 2\OPT(V)$.
Let $\mathcal{A}(V',W')$ be the output distribution of Algorithm~\ref{alg:knapsack} on the instance $V'$ with all the weights being divided by $W'$. Then, we have
\[
    \E[v(\hat{\mathcal{A}})] \geq \E[v(\mathcal{A}(V',W'))]\delta
    \geq (1-\epsilon')\OPT(V',W')\delta
    \geq (1-4\epsilon)(1-\epsilon')\OPT(V,W)
    \geq (1-4\epsilon-\epsilon')\OPT(V,W),
\]
where the first inequality is from Line~\ref{line:roundd} of Algorithm~\ref{alg:fptas}, the second inequality is from the assumption of the lemma, and the third inequality is from~\eqref{eq:m2e}. Therefore the lemma is proved.
\end{proof}

Let $\hat{\mathcal{A}}^i$ be the output distribution of Algorithm~\ref{alg:fptas} on the instance $V\setminus \{i\}$. For $\delta>0$, let $\hat{\mathcal{A}}(\delta)$ (resp., $\hat{\mathcal{A}}^i(\delta)$) be the output distribution of Algorithm~\ref{alg:fptas} on the instance $V$ (resp., $V\setminus \{i\}$), conditioned on that the algorithm chose this $\delta$ at Line~\ref{line:choosef}.
To bound the earth mover's distance, we consider transporting the probability mass in such a way that the mass of $\hat{\mathcal{A}}$ corresponding to a particular choice of $\delta$ is transported to that for $\mathcal{A}^i$ with the same $\delta$ (as far as we can).
For the same choice of $\delta$, we consider transporting the probability mass in the same way as in the analysis of Algorithm~\ref{alg:knapsack}.
The remaining mass is transported arbitrarily.

Now, we bound the average sensitivity of Algorithm~\ref{alg:fptas} by bounding the contribution to the earth mover's distance due to the difference between the distributions of $\delta$ and $\delta^i$ by a similar argument to that for Lemma~\ref{lem:errorbyc}. We have the following.
\begin{lemma}\label{lem:fpsens}
Let $s$ be the average sensitivity of Algorithm~\ref{alg:knapsack}. Then, Algorithm~\ref{alg:fptas} has average sensitivity $s+2$.
\end{lemma}
\begin{proof}
For $i\in V$, we have
\begin{align}
    &\EM(\hat{\mathcal{A}},\hat{\mathcal{A}}^i) \nonumber \\
    &\leq \frac{n}{\FOPT(V)\epsilon}\left(\int_{\FOPT(V)\epsilon/n}^{2\FOPT(V\setminus \{i\})\epsilon/n}\EM(\mathcal{A}(\delta),\mathcal{A}^i(\delta)) \mathrm{d}\delta
    + \int_{2\FOPT(V\setminus \{i\})\epsilon/n}^{2\FOPT(V)\epsilon/n}n \mathrm{d}\delta\right) \nonumber \\
    &\leq \frac{n}{\FOPT(V)\epsilon}\int_{\FOPT(V)\epsilon/n}^{2\FOPT(V)\epsilon/n}\EM(\mathcal{A}(\delta),\mathcal{A}^i(\delta)) \mathrm{d}\delta
     + \frac{n}{\FOPT(V)}\cdot 2(\FOPT(V)-\FOPT(V\setminus \{i\}))\label{eq:fpdiv}.
\end{align}
Here, the first inequality is from the way we transport the probability mass and
\[
    \frac{n}{\FOPT(V)\epsilon}\leq \frac{n}{\FOPT(V\setminus \{i\})\epsilon},
\]
which follows from $\FOPT(V)\geq \FOPT(V\setminus \{i\})$. The second inequality is from $\FOPT(V)\geq \FOPT(V\setminus \{i\})$.
Now, we have
\begin{align*}
    &\frac{1}{n}\sum_{i\in V}\EM(\hat{\mathcal{A}},\hat{\mathcal{A}}^i)\\
    &\leq \frac{1}{\FOPT(V)\epsilon}\int_{\FOPT(V)\epsilon/n}^{2\FOPT(V)\epsilon/n}\sum_{i\in V}\EM(\mathcal{A}(\delta),\mathcal{A}^i(\delta)) \mathrm{d}\delta + \frac{1}{\FOPT(V)}\cdot 2\sum_{i\in V}(\FOPT(V)-\FOPT(V\setminus \{i\}))\\
    &\leq \frac{1}{\FOPT(V)\epsilon}\int_{\FOPT(V)\epsilon/n}^{2\FOPT(V)\epsilon/n}s \mathrm{d}\delta +2\\
    &= s+2,
\end{align*}
where the first inequality is from~\eqref{eq:fpdiv} and the second inequality is from the definition of $s$ and Lemma~\ref{lem:fracdiffsum}. Therefore the lemma is proved.
\end{proof}


\begin{proof}[Proof of Theorem~\ref{thm:intro}]
Combining Theorem~\ref{thm:main} and Lemmas~\ref{lem:fpapprox} and~\ref{lem:fpsens}, and replacing $\epsilon$ and $\epsilon'$ by $\epsilon/5$ proves the claim.
\end{proof}

\section{Lower Bound}\label{sec:lower_bound}

In this section, we prove that any $(1-\epsilon)$-approximation algorithm for the knapsack problem has $\Omega(\epsilon^{-1})$ average sensitivity, which shows that the upper bound of $O(\epsilon^{-1}\log \epsilon^{-1})$ (Theorem~\ref{thm:main}) is tight up to a logarithmic factor in $\epsilon^{-1}$.

We prove the lower bound by constructing an instance $V$ that requires average sensitivity $\Omega(\epsilon^{-1})$.
Let $k=\lfloor\epsilon^{-1}/8\rfloor$.
Then, let $V_1$ be the set of $k$ copies of an item with weight $\frac{1}{k}$ and value $1$, $V_2$ be the set of $k-1$ copies of an item with weight $\frac{1}{k-1}$ and value $\frac{2k-1}{2k-2}$, and $V=V_1\cup V_2$.
Note that the feasible domain of the instance $V$ is
\[
    \{X\subseteq V:|X|\leq k-1\}\cup \{V_1\}.
\]
We fix a $(1-\epsilon)$-approximation algorithm and $\mathcal{A}$ (resp., $\mathcal{A}^i$) be its output distribution on the instance $V$ (resp., $V\setminus \{i\}$). The next lemma is clear from the construction of the instance.
\begin{lemma}\label{lem:vopt}
    The instance $V$ has a unique optimal solution $V_1$, which has value $k$.
    Any other feasible solution has value at most $k-\frac{1}{2}$.
\end{lemma}

Therefore we have the following.
\begin{lemma}\label{lem:out1}
Any (randomized) $(1-\epsilon)$-approximation algorithm that runs on $V$ outputs $V_1$ with probability at least $\frac{3}{4}$.
\end{lemma}
\begin{proof}
Let $p$ be the probability that the algorithm outputs $V_1$.
Then, we have
\begin{align*}
    \E[v(\mathcal{A})]\leq pk+(1-p)\left(k-\frac{1}{2}\right)
    = \left(1-\frac{1-p}{2k}\right)k
    \leq \left(1-4(1-p)\epsilon\right)k
    =\left(1-4(1-p)\epsilon\right)\OPT(V),
\end{align*}
where the first and last equalities are from Lemma~\ref{lem:vopt} and the second inequality is from $k=\lfloor\epsilon^{-1}/8\rfloor$.
Therefore, to obtain an approximation ratio at least $1-\epsilon$, we must have $p \geq \frac{3}{4}$.
\end{proof}

Let $V'_1$ be the set obtained from $V_1$ by removing an arbitrary item, and let $V'=V'_1\cup V_2$.
The next lemma is clear from the fact that the feasible domain of the instance $V'$ is $\{X\subseteq V:|X|\leq k-1\}$.
\begin{lemma}\label{lem:vdopt}
    The instance $V'$ has a unique optimal solution $V_2$, which has value $k-\frac{1}{2}$.
    Also, any feasible solution with less than $\frac{k-1}{2}$ items in $V_2$ has a value less than $k-\frac{3}{4}$.
\end{lemma}

Therefore we have the following.
\begin{lemma}\label{lem:out2}
Any randomized $(1-\epsilon)$-approximation algorithm that runs on $V'$ outputs a set that contains at least $\frac{k-1}{2}$ items in $V_2$ with probability at least $\frac{1}{2}$.
\end{lemma}
\begin{proof}
Let $p$ be the probability that the algorithm outputs the set with at least $\frac{k-1}{2}$ items in $V_2$. Then,
\begin{align*}
    \E[\mathcal{A}]
    & \leq p\left(k-\frac{1}{2}\right)+(1-p)\left(k-\frac{3}{4}\right)
    = \left(1-\frac{1-p}{4\left(k-\frac{1}{2}\right)}\right)\left(k-\frac{1}{2}\right) \\
    &\leq \left(1-2(1-p)\epsilon\right)\left(k-\frac{1}{2}\right)
    = \left(1-2(1-p)\epsilon\right)\OPT(V'),
\end{align*}
where the first and last equalities is from Lemma~\ref{lem:vdopt}, and the last inequality is from $k=\lfloor\epsilon^{-1}/8\rfloor$.
Therefore, to obtain an approximation ratio at least $1-\epsilon$, we must have $p \geq \frac{1}{2}$.
\end{proof}

Combining these lemmas yields the following.

\begin{proof}[Proof of Theorem~\ref{thm:intro-lower-bound}]
We evaluate the earth mover's distance between $\mathcal{A}$ and $\mathcal{A}^i$, where $i \in V_1$.
By Lemmas~\ref{lem:out1} and~\ref{lem:out2}, a probability mass of at least $\frac{1}{4}$ must be transported from the solution $V_1$ for the instance $V$ to the solutions for $V\setminus \{i\}$, which have at least $\frac{k-1}{2}$ items in $V_2$.
Therefore, we have
\begin{align}
    \frac{1}{|V|}\sum_{i\in V}\EM(\mathcal{A},\mathcal{A}^i)
    \geq \frac{1}{|V|}\cdot |V_1| \cdot \frac{1}{4} \cdot \frac{k-1}{2} = \frac{k(k-1)}{8(2k-1)} = \Omega(\epsilon^{-1}),
\end{align}
where the last inequality is from $k=\lfloor\epsilon^{-1}/8\rfloor$.
\end{proof}

\section{Simple Knapsack}\label{sec:simple}

In this section, we present a deterministic algorithm for the simple knapsack problem with a small average sensitivity.
Recall that the knapsack problem is \emph{simple} if the value of an item is proportional to its weight.
Without loss of generality, we assume that $W=1$ and $w(i)=v(i)$ holds for every item $i$.


\begin{algorithm}[t!]
\caption{Stable-on-Average Deterministic Algorithm for the Simple Knapsack Problem}\label{alg:simple}
Let $L$ be the set of items with weight at least $\epsilon$\;
Let $R\leftarrow \OPT(L)$\;
Reorder the items so that $V\setminus L = \{1,\ldots,|V\setminus L|\}$ and $w(1)\leq \cdots \leq w(|V\setminus L|)\leq \epsilon$\;
Let $k$ be the largest index with $w(1)+\cdots+w(k)\leq 1-w(R)$\;
\Return $R\cup \{1,\dots, k\}$.
\end{algorithm}

Our algorithm is given in Algorithm~\ref{alg:simple}.
To calculate $\OPT(L)$, we apply brute-force search over the subsets $R\subseteq L$ with $|R|\leq \epsilon^{-1}$.
The following is clear from the design of the algorithm.
\begin{lemma}\label{lem:simple-approx}
If $w(V)\leq 1$, Algorithm~\ref{alg:simple} outputs $V$.
Otherwise, Algorithm~\ref{alg:simple} outputs the set $R\cup \{1,\dots, k\}$ with $w(R\cup \{1,\dots, k\})\geq 1-\epsilon$.
\end{lemma}
Therefore, Algorithm~\ref{alg:simple} has an approximation ratio $1-\epsilon$.
We now analyze the average sensitivity.
\begin{lemma}\label{lem:simple-sensitivity}
Algorithm~\ref{alg:simple} has an average sensitivity $\epsilon^{-1}+2$.
\end{lemma}
\begin{proof}
Let $A$ (resp., $A^i$) be the output of the algorithm~\ref{alg:simple} on the instance $V$ (resp., $V\setminus \{i\}$).
Let $S\subseteq L$ be a set that attains $\OPT(L)$.
If $i\in S$, then we trivially have $|A\triangle A^i|\leq n$.
If $i\in L\setminus S$, then we have $|A\triangle A^i| = 0$.
Let $k$ be the index such that $A=\OPT(L)\cup\{1,\dots, k\}$.
If $i\in V\setminus L$ with $i > k$, then we have $|A\triangle A^i| = 0$.

Suppose $i\in V\setminus L$ with $i\leq k$.
We have
\begin{align*}
    &\OPT(L,1)+w(1)+\cdots+w(i-1)+w(i+1)+\cdots+w(k+2)\\
    &\geq \OPT(L,1)+w(1)+\cdots+w(i-1)+w(i)+w(i+1)+\cdots+w(k+1)>1,
\end{align*}
where the first inequality is from $w(i)\leq w(k+2)$ and the second inequality is from the design of the algorithm.
Therefore, we have $A\triangle A^i\subseteq \{i,k+1\}$ and $|A\triangle A^i|\leq 2$.

Combining the bounds above, we have
\[
    \frac{1}{n}\sum_{i\in V}|A\triangle A^i|\leq \frac{1}{n}(|S|\cdot n+n\cdot 2)\leq \epsilon^{-1}+2,
\]
where the last inequality is from $|S|\leq \epsilon^{-1}$.
\end{proof}
Theorem~\ref{thm:intro-simple} follows by combining Lemmas~\ref{lem:simple-approx} and~\ref{lem:simple-sensitivity}.

\section{Applications to the Dynamic Setting}\label{sec:recourse}

In this section, we discuss an application of our result to the incremental knapsack problem in the random-order setting.
In this problem, items arrive one by one in a random order from an empty set of items, and the goal is to maintain a good solution for the current set of items.




We show that an algorithm with bounded average sensitivity implies a dynamic algorithm with bounded amortized recourse.
Let $U$ be the ground set. For $i\in U$, let $\mathcal{S}$ and $\mathcal{S}^i$ be two distributions over subsets of $U$ and $U\setminus \{i\}$, respectively.
We assume that $\frac{1}{|U|}\sum_{i\in U}\EM\left(\mathcal{S},\mathcal{S}^i\right)\leq s$ holds. Then, we say that the probability transportation to achieve the average sensitivity $s$ is \emph{one-way computable} if the following holds:
\begin{itemize}
    \item For each $i$, there is a distribution of a pair of sets $\mathcal{D}^i$ such that its marginal distributions on the first and second coordinates are equal to $\mathcal{S}$ and $\mathcal{S}^i$, respectively.
    \item $\displaystyle \frac{1}{|U|}\sum_{i\in U}\E_{(S,S^i)\sim \mathcal{D}^i}\left[S\triangle S^i\right]\leq s$. 
    \item For each $i$, there is a (randomized) algorithm that takes a set $S^i$ and returns a set $S$, such that if the input $S^i$ is sampled from the distribution $\mathcal{S}^i$, the distribution of the pair $(S,S^i)$ is equal to $\mathcal{D}^i$. 
\end{itemize}
Recalling the proofs in Section~\ref{sec:general}, for any $U\subseteq V$ and $\epsilon>0$, we can see that the probability transportation that achieves the average sensitivity $O(\epsilon^{-1}\log \epsilon^{-1})$ is one-way computable.

The following gives a general transformation from algorithms with small average sensitivity to dynamic algorithms with small amortized recourse.
\begin{proposition}\label{prop:boundedrecourse}
Assume that there is a (randomized) algorithm $A$ with approximation ratio $1-\epsilon$ and average sensitivity $s$ such that the probability transportation to achieve average sensitivity $s$ is one-way computable. Then, there is an algorithm for the incremental dynamic knapsack problem in the random-order setting such that
\begin{itemize}
    \item the expected approximation ratio is $1-\epsilon$ for each step, where the expectation is taken over the random bits used in the algorithm, and
    \item the expected amortized recourse is at most $s$, where the expectation is taken over both the random bits used in the algorithm and the arrival order of the input.
\end{itemize}
\end{proposition}
\begin{proof}
Let $\sigma$ be a random permutation of $\{1,\dots, |V|\}$ and let $X_{\sigma,0}=\emptyset$. For each $k=1,\dots, |V|$, let $X_{\sigma, k}$ be the output of $A$ for the instance $\{\sigma(1),\dots, \sigma(k)\}$, which is computed from $X_{\sigma,k-1}$ by the probability transportation in the statement.
Then, we see that the expectation of approximation ratio is $1-\epsilon$ because the distribution of $X_{\sigma, k}$ is the same as the output distribution of $A$ on $\{\sigma(1),\dots, \sigma(k)\}$.
The expected amortized recourse is then bounded as
\begin{align*}
    \E_{\sigma}\left[\frac{1}{n}\sum_{k=1}^{n}|X_{\sigma, k-1}\triangle X_{\sigma, k}|\right]
    &=\E_{\sigma}\left[\frac{1}{n}\sum_{k=1}^{n}\EM\left(\mathcal{A}(\{\sigma(1),\dots, \sigma(k-1)\}),\mathcal{A}(\{\sigma(1),\dots, \sigma(k)\})\right)\right]\\
    &\leq \frac{1}{n}\sum_{k=1}^{n}\E_{(i,S)\colon i\in S, |S|=k}\left[\EM\left(\mathcal{A}(S),\mathcal{A}(S\setminus \{i\})\right)\right]\\
    &\leq \frac{1}{n}\sum_{k=1}^{n}s=s.
    \qedhere
\end{align*}
\end{proof}
The above proposition and our result in Section~\ref{sec:general} imply that there is an $(1-\epsilon)$-approximation algorithm for the incremental knapsack problem in the random-order setting with amortized recourse $O(\epsilon^{-1} \log \epsilon^{-1})$.



We can slightly modify the algorithm in Section~\ref{sec:general} to bound the update time.
As a preprocess, we apply a standard argument on approximating the knapsack problem:
We round the value of each item $i$ into the largest integer power of $(1-\epsilon)^{-1}$ that does not exceed $v(i)$.
Because this operation multiples the value of each item by at least $1-\epsilon$, the optimal value of the knapsack problem is also multiplied by at least $1-\epsilon$.

We can see that, the probability transportation that achieves average sensitivity $O(\epsilon^{-1}\log \epsilon^{-1})$ can be computed from the following information:
\begin{itemize}
    \item For two real values $t_1<t_2$, the set of items with the smallest weight that satisfies $t_1\leq v(A_t) < t_2$, where all items have value at least $\epsilon\cdot \FOPT(V)$ (Line~\ref{line:smallestweight} of Algorithm~\ref{alg:knapsack})
    \item For items with $v(1)/w(1)\geq \cdots \geq v(n)/w(n)$ and a weight limit $W$, the last index $k$ such that $v(1)+\cdots +v(k)\leq W$ (Line~\ref{line:defx} of Algorithm~\ref{alg:knapsack} and Line~\ref{line:smalltaket} of Algorithm~\ref{alg:smallknapsack})
\end{itemize}

Let us consider computing the first one.
Because any feasible solution can contain at most $\epsilon^{-1}$ items, we only need to consider the $\epsilon^{-1}$ lightest items for each value candidate.
As there are $\log_{(1-\epsilon)^{-1}}\epsilon^{-1}$ value candidates, we can simply compute the desired set by brute-force search over the set of constant cardinality.
The second is calculated in $O(\log n)$ time by maintaining the set of items using a binary search tree.
Therefore, we can compute the probability transportation in $O(f_\epsilon+\log n)$ time for some $f_\epsilon$.
Hence we obtain Theorem~\ref{thm:dynamic}.

All results in this section can also be applied to the \emph{decremental knapsack problem}, where the input is initially a ground set $V$ and items are deleted one by one at each step.




\bibliography{bib}

\begin{thebibliography}{10}

\bibitem{eppstein1999dynamic}
Mikhail Atallah and Marina Blanton.
\newblock {\em Algorithms and theory of computation handbook}.
\newblock CRC press, 2009.

\bibitem{babaioff2007knapsack}
Moshe Babaioff, Nicole Immorlica, David Kempe, and Robert Kleinberg.
\newblock A knapsack secretary problem with applications.
\newblock In {\em Approximation, randomization, and combinatorial optimization. Algorithms and techniques}, pages 16--28. Springer, 2007.

\bibitem{bohm2020fully}
Martin B{\"o}hm, Franziska Eberle, Nicole Megow, Bertrand Simon, Lukas N{\"o}lke, Jens Schl{\"o}ter, and Andreas Wiese.
\newblock Fully dynamic algorithms for knapsack problems with polylogarithmic update time.
\newblock In {\em IARCS Annual Conference on Foundations of Software Technology and Theoretical Computer Science}, 2021.

\bibitem{buchbinder2009online}
Niv Buchbinder and Joseph Naor.
\newblock Online primal-dual algorithms for covering and packing.
\newblock {\em Mathematics of Operations Research}, 34(2):270--286, 2009.

\bibitem{chan2018approximation}
Timothy Chan.
\newblock Approximation schemes for 0-1 knapsack.
\newblock In {\em Symposium on Simplicity in Algorithms}, 2018.

\bibitem{dantzig1957discrete}
George Dantzig.
\newblock Discrete variable extremum problems.
\newblock {\em Operations Research}, 5:266--277, 1957.

\bibitem{dwork2006calibrating}
Cynthia Dwork, Frank McSherry, Kobbi Nissim, and Adam Smith.
\newblock Calibrating noise to sensitivity in private data analysis.
\newblock In {\em Theory of Cryptography Conference}, pages 265--284. Springer, 2006.

\bibitem{feldkord2018fully}
Bj{\"o}rn Feldkord, Matthias Feldotto, Anupam Gupta, Guru Guruganesh, Amit Kumar, S{\"o}ren Riechers, and David Wajc.
\newblock Fully-dynamic bin packing with little repacking.
\newblock In {\em International Colloquium on Automata, Languages, and Programming}, 2018.

\bibitem{gupta2017online}
Anupam Gupta, Ravishankar Krishnaswamy, Amit Kumar, and Debmalya Panigrahi.
\newblock Online and dynamic algorithms for set cover.
\newblock In {\em ACM SIGACT Symposium on Theory of Computing}, pages 537--550, 2017.

\bibitem{gupta2020fully}
Anupam Gupta and Roie Levin.
\newblock Fully-dynamic submodular cover with bounded recourse.
\newblock {\em IEEE Symposium on Foundations of Computer Science}, pages 1147--1157, 2020.

\bibitem{gupta2013fully}
Manoj Gupta and Richard Peng.
\newblock Fully dynamic $(1 + \epsilon)$-approximate matchings.
\newblock In {\em IEEE Symposium on Foundations of Computer Science}, pages 548--557, 2013.

\bibitem{han2019online}
Xin Han, Yasushi Kawase, Kazuhisa Makino, and Haruki Yokomaku.
\newblock Online knapsack problems with a resource buffer.
\newblock In {\em International Symposium on Algorithms and Computation}, 2019.

\bibitem{ibarra1975fast}
Oscar Ibarra and Chul Kim.
\newblock Fast approximation algorithms for the knapsack and sum of subset problems.
\newblock {\em Journal of the ACM}, 22(4):463--468, 1975.

\bibitem{iwama2002removable}
Kazuo Iwama and Shiro Taketomi.
\newblock Removable online knapsack problems.
\newblock In {\em International Colloquium on Automata, Languages, and Programming}, pages 293--305, 2002.

\bibitem{jin2019improved}
Ce~Jin.
\newblock An improved fptas for $0$-$1$ knapsack.
\newblock In {\em International Colloquium on Automata, Languages, and Programming}, 2019.

\bibitem{karp1972reducibility}
Richard Karp.
\newblock Reducibility among combinatorial problems.
\newblock In {\em Complexity of Computer Computations}, pages 85--103. Springer, 1972.

\bibitem{kellerer1999new}
Hans Kellerer and Ulrich Pferschy.
\newblock A new fully polynomial time approximation scheme for the knapsack problem.
\newblock {\em Journal of Combinatorial Optimization}, 3(1):59--71, 1999.

\bibitem{kellerer2004improved}
Hans Kellerer and Ulrich Pferschy.
\newblock Improved dynamic programming in connection with an fptas for the knapsack problem.
\newblock {\em Journal of Combinatorial Optimization}, 8(1):5--11, 2004.

\bibitem{kellerer2004knapsack}
Hans Kellerer, Ulrich Pferschy, and David Pisinger.
\newblock {\em Knapsack problems}.
\newblock Springer, 2004.

\bibitem{korte2011combinatorial}
Bernhard Korte and Jens Vygen.
\newblock {\em Combinatorial optimization}, volume~1.
\newblock Springer, 2011.

\bibitem{kumabe2022average}
Soh Kumabe and Yuichi Yoshida.
\newblock Average sensitivity of dynamic programming.
\newblock In {\em {ACM}-{SIAM} Symposium on Discrete Algorithms}, pages 1925--1961, 2022.

\bibitem{lawler1979fast}
Eugene Lawler.
\newblock Fast approximation algorithms for knapsack problems.
\newblock {\em Mathematics of Operations Research}, 4(4):339--356, 1979.

\bibitem{marchetti1995stochastic}
Alberto Marchetti-Spaccamela and Carlo Vercellis.
\newblock Stochastic on-line knapsack problems.
\newblock {\em Mathematical Programming}, 68(1):73--104, 1995.

\bibitem{mcsherry2007mechanism}
Frank McSherry and Kunal Talwar.
\newblock Mechanism design via differential privacy.
\newblock In {\em IEEE Symposium on Foundations of Computer Science}, pages 94--103, 2007.

\bibitem{Murai2019}
Shogo Murai and Yuichi Yoshida.
\newblock Sensitivity analysis of centralities on unweighted networks.
\newblock In {\em The World Wide Web Conference}, page 1332–1342, 2019.

\bibitem{onak2010maintaining}
Krzysztof Onak and Ronitt Rubinfeld.
\newblock Maintaining a large matching and a small vertex cover.
\newblock In {\em ACM Symposium on Theory of Computing}, pages 457--464, 2010.

\bibitem{Peng2020}
Pan Peng and Yuichi Yoshida.
\newblock Average sensitivity of spectral clustering.
\newblock In {\em {ACM} {SIGKDD} International Conference on Knowledge Discovery {\&} Data Mining}, pages 1132--1140. {ACM}, 2020.

\bibitem{rhee2015faster}
Donguk Rhee.
\newblock {\em Faster fully polynomial approximation schemes for knapsack problems}.
\newblock PhD thesis, Massachusetts Institute of Technology, 2015.

\bibitem{shiloach1981line}
Yossi Shiloach and Shimon Even.
\newblock An on-line edge-deletion problem.
\newblock {\em Journal of the ACM}, 28(1):1--4, 1981.

\bibitem{vadhan2017complexity}
Salil Vadhan.
\newblock The complexity of differential privacy.
\newblock In {\em Tutorials on the Foundations of Cryptography}, pages 347--450. Springer, 2017.

\bibitem{van2019dynamic}
Jan van~den Brand and Danupon Nanongkai.
\newblock Dynamic approximate shortest paths and beyond: Subquadratic and worst-case update time.
\newblock In {\em IEEE Symposium on Foundations of Computer Science}, pages 436--455, 2019.

\bibitem{Varma2021}
Nithin Varma and Yuichi Yoshida.
\newblock Average sensitivity of graph algorithms.
\newblock In {\em {ACM}-{SIAM} Symposium on Discrete Algorithms}, pages 684--703, 2021.

\bibitem{Yoshida2021}
Yuichi Yoshida and Samson Zhou.
\newblock Sensitivity analysis of the maximum matching problem.
\newblock In {\em Innovations in Theoretical Computer Science}, pages 58:1--58:20, 2021.

\bibitem{zhou2008budget}
Yunhong Zhou, Deeparnab Chakrabarty, and Rajan Lukose.
\newblock Budget constrained bidding in keyword auctions and online knapsack problems.
\newblock In {\em International Workshop on Internet and Network Economics}, pages 566--576. Springer, 2008.

\end{thebibliography}
\newpage
\appendix




\end{document}